\documentclass[11pt,a4paper,leqno]{article}

\usepackage[english]{babel}
\usepackage{latexsym}
\usepackage{amsmath, amssymb}
\usepackage{graphicx}
\usepackage{epsfig}
\usepackage{dsfont}
\usepackage[latin1]{inputenc} % entree 8 bits iso-latin1
\usepackage[T1]{fontenc}      % encodage 8 bits des fontes utilisees

\usepackage{latexsym,dsfont,textcomp,amsmath} %natbib
\usepackage{color}
\usepackage{stmaryrd}
\usepackage{fancyhdr}
\usepackage{float}

\title{The impact of uncertainties on the pricing of contingent claims}

\author{Simone Scotti\footnote{Institut de Mathematique,
Ecole Polythecnique Federale de Lausanne, Station 8, CH-1015
Lausanne, EPFL, Switzerland. Email: simone.scotti@epfl.ch}}

\newtheorem{teorema}{Theorem}[section]

\newtheorem{proposizione}[teorema]{Proposition}
\newtheorem{corollario}[teorema]{Corollary}

\newtheorem{remarque}{Remark}[section]
\newtheorem{definizione}{Definition}[section]
\newtheorem{assumption}{Assumption}

\newtheorem{principle}{Principle}

\newenvironment{proof}{{\textbf Proof:} }{\smallskip \begin{flushright}
$\Box$
\end{flushright}}

\begin{document}

\maketitle

\begin{abstract}
  We study the effect of parameters uncertainties on a stochastic diffusion 
  model, in particular the impact on the pricing of contingent claims, thanks to Dirichlet Forms methods.
  We apply recent techniques, developed by Bouleau, to hedging procedures in order to compute  the sensitivities  of SDE trajectories with respect to parameter perturbations.  We show that this model can reproduce a bid-ask spread.
  We also prove that, if the stochastic differential equation admits a closed form representation,  
  also the sensitivities have closed form representations.
  
  We exhibit the case of log-normal diffusion and we show that this framework foresees a smiled 
  implied volatility surface coherent with historical data. 

\vspace{0.5cm}

\vspace{0.5cm}
 { \bf AMS:} 60H30, 91B16 and 91B70

  %\JEL{G13 \and C02}

\end{abstract}

\section{Introduction}

The purpose of this article is to study the sensitivity of the flow of a stochastic differential 
equation with respect to its parameters and, in particular, the adjustments on the computation 
of conditional  expectation.
The main application is to evaluate the correction of option pricing due to 
the uncertainty on parameters driving the SDE that describe the prices.

The usual framework of asset pricing is due to Black and Scholes  \cite{bib:Black-Scholes}, 
it provides conditions on the market which guarantee the existence of a hedging strategy to cover a contingent claim and, therefore, the existence and the uniqueness of its price. However, the hypotheses of classical asset pricing have clearly some deficiencies and drawback: they assume in particular a complete knowledge of the market and, in particular, of the parameters driving the SDE that describes the assets prices. 

The risk assessment is often carry out in terms of sensitivity with respect to a deterministic variation of model parameters. These sensitivities are usually called greeks. 
However, this approach encounters mathematical difficulties, in dealing with
infinite dimensions, and practical ones.
A large literature exists in the subject to cover the risk of uncertainty. The most important example is the uncertain volatility model,  see Avelaneda et al. \cite{bib:Avelaneda}, 
 Lyons \cite{bib:Lyons} and Barles \cite{bib:Barles}, which takes into account the difficulties to calibrate the volatility in the Black Schoels model. Their strategy is to apply the stochastic control to super-hedge a contingent claim under model uncertainty. We also recall the recent paper of Denis and Martini \cite{bib:Denis-1} where a more general framework is introduced to study super-replication and model uncertainty, their methodology make use of the Choquet capacities.

The principal drawback of the super-replication techniques is that the super-hedging cost is too high and the corresponding super-hedging strategy is too conservative. 
As an example, Kramkov \cite{bib:Kramkov} shows that the super-hedging cost corresponds to the value of the option under the least favorable martingale measure. Likewise, Bellamy and Jeanblanc \cite{bib:Bellamy} prove that in a jumps-diffusion model the price range for a call option  
corresponds to the interval given by the no-arbitrage conditions, e.g. the super-hedging strategy of a call option is to buy one unit of the underlying. 
It seems therefore that super-hedging is not an effective methodology, since the option buyer will prefer to buy the underlying rather than pay the same price to have an option. The option writer has to take some risk in order to propose a competitive price. In other words, the option writer has to take the risk to lose money, with a small probability that this event occurs, in order to have the opportunity to make money with the contract, i.e. selling the option and hedging it. 

Many attempts to attack the problem are appeared in literature, Avelaneda et al. 
\cite{bib:Avelaneda} bound the volatility between two levels.  Other ways are the pricing via utility maximization, see for instance El Karoui and Rouge \cite{bib:ElKaroui}, or via local risk minimization, see for instance Follmer and Schweizer \cite{bib:Follmer}.

An alternative way, based on potential theory 
(for reference see Albeverio \cite{bib:Albeverio}, Bouleau and Hirsch
\cite{bib:Bouleau-Hirsch} and Fukushima et al.
\cite{bib:Fukushima}), has been suggested by Bouleau
\cite{bib:Bouleau-erreur}, this method is an extension of Malliavin calculus. 
We assume that the uncertainties are supposed to be infinitesimal random variables in
 accord with statistical estimation. These random variables are generally poorly known, 
 so we assume that only the two first order are known, i.e. the biases and the 
 variances/covariances, since the accuracy is usually computed using statistical measures, like Fisher information. 
If the uncertainty on parameter is small, we may neglect  orders
higher than the second, i.e. we take into account only the  bias and the variance of the parameters. 
This cut is justified by the evidence that the parameters estimation is often hard, owing to 
the non-linearity, the complexity of the equations and to the bad quality or the shortage 
of the data.
In error theory using Dirichlet forms, we associate
the variance to the ``carr\'{e} du champ'' operator $\Gamma$ and the
shift to the generator of the semigroup $\mathcal{A}$. The two operators are closed in a suitable space and have a closed chain rule, see Bouleau  \cite{bib:Bouleau-erreur2}.

The main result of this paper gives closed forms for the corrections on the price of contingent claims due to the uncertainties on the parameters. These are corrections for an investor who searches to compensate the systematic bias and that accepts a residual risk on the parameters sensitivity. We define this  risk taking in a statistical way.
A direct consequence of this result is the separation of the buy and sell prices due 
to the asymmetry between the buyer and the seller with respect to the risk on parameters.
An outstanding result is that a systematic bias exists even if the parameters are unbiased, 
this bias is due to the non-linearity of the payoff.
Our analysis covers a large class of stochastic process with continuous path. In particular, equity models with local or stochastic volatility. 

We exhibit, at the end, an example in the case of Black Scholes model with uncertainty on 
volatility parameter. We show that that our analysis foresees a smiled implied volatility and
We give the explicit formula for vanilla prices and the bid-ask spread. 

The paper is organized as follows:
In section 2, we present a survey of error theory using Dirichlet forms technique. 
Section 3 is devoted to the study of the impact of uncertainty on a diffusion model. 
We analyze the profit and loss process and we compute its law depending both on
the underlying diffusion and the parameter uncertainty. We also introduce a 
pricing principle to over-hedge the contingent claim and we exhibit the bid and ask prices.
In section 4, we give an example with log-normal diffusion without drift. We exhibit
the bid and the ask prices and we prove, under some hypotheses, that the implied volatility has
a smiled behavior.

\section{Mathematical Framework for Uncertainty}

We begin by giving a general introduction to the study of the sensitivity with respect to a stochastic perturbation and a formal definition of the framework that we will use in accord with the error theory using Dirichlet forms in accord with  Bouleau \cite{bib:Bouleau-erreur}. In this survey, we follow \cite{bib:Bouleau-erreur2}.

We consider a function $F(U_1,\, U_2, \, ...)$, depending on parameters
 $(U_1,\, U_2, \, ...)$ that we suppose afflicted with uncertainties. We assume that the function $F$
 is regular enough and we search to evaluate the impact on $F$ of the uncertainties on $U_i$. 
 The first study of this problem is due to Gauss that prove the following expansion for the variance of $F$ if the uncertainties are small compared with the parameters values:
 $$
 \text{Var}[F] = \sum_{i,\, j} \frac{\partial F}{\partial U_i} \,  \frac{\partial F}{\partial U_j} \, \text{Var}[U_i,\, U_j] 
 $$
However, this relation is proved only if the number of  the parameters is fixed and if $F$ has an explicit formula, i.e. if it does not be defined via a limit. The main aim of the error theory using Dirichlet forms is to go beyond these limits in order to evaluate the effect of a stochastic fluctuation on more complex objects like stochastic integrals. For that, we study a parameter $x$ with a small uncertainty $\sqrt{\epsilon} \,Y$, on which we compute a non-linear function $f$. Our parameter $x$ is replaced by the random variable $X=x + \sqrt{\epsilon} \,Y$ with variance $\epsilon \,\text{Var}[Y] $. We now apply the Taylor expansion to function $f$ and we find
\begin{eqnarray*}
\text{Bias}[f(X)] & = & \mathbb{E}[f(X)-f(x)] = f^{\prime}(x) \; \text{Bias}[X] + \frac{1}{2}
 f^{\prime \prime}(x) \; \text{Var}[X] + o(\epsilon) \\
\text{Var}[f(X)] & = & \mathbb{E}\left[ \{f(X)-f(x)\}^2 \right] = \left[f^{\prime}(x)\right]^2 \; \text{Var}[X] 
+ o(\epsilon). 
\end{eqnarray*} 
If we suppose $\epsilon$ really small we can cut the high terms on $\epsilon$ and find two closed chain rules for the bias and the variance. These are known in literature since the variance follows the same rule of a carr\'{e} du champ operator of a probability space equipped with a local Dirichlet form, while the bias verifies the rule of the generator of the semigroup associated to the Dirichlet form, see for instance Bouleau and Hirsch \cite{bib:Bouleau-Hirsch}. 
The main advantage of this comparison is that the carr\'{e} du champ  and the generator of semigroup are closed  operator with respect to the graph norm, see for instance Fukushima et al. 
\cite{bib:Fukushima}. Therefore, they are good operators to study objects defined by limits, like stochastic integrals.

The axiomatization of this idea is introduced by Bouleau \cite{bib:Bouleau-erreur} as follows: He defines an error structure as a probability space equipped with a local Dirichlet form owning a  
carr\'{e} du champ.

\begin{definizione}[Error structure]\hfill

An error structure is a term
$\displaystyle  \left( \widetilde{\Omega}, \, \widetilde{\mathcal{F}}, \, \widetilde{\mathbb{P}}, \, \mathbb{D}, \, \Gamma \right)$, where

\begin{itemize}
\item{$\left( \widetilde{\Omega}, \, \widetilde{\mathcal{F}}, \,
\widetilde{\mathbb{P}} \right)$ is a probability space;}
 \item{$\mathbb{D}$ is a dense sub-vector space of $L^2\left(
\widetilde{\Omega}, \, \widetilde{\mathcal{F}},
      \, \widetilde{\mathbb{P}}\right)$;}
\item{$\Gamma$ is a positive symmetric bilinear application from
$\mathbb{D} \, \times
    \, \mathbb{D}$ into $L^1 \left( \widetilde{\Omega}, \, \widetilde{\mathcal{F}},
      \, \widetilde{\mathbb{P}} \right)$ satisfying the functional calculus of class
    $\mathcal{C}^1 \cap Lip$, i.e. if F and G are of class $\mathcal{C}^1$ and Lipschitzian, u and
v $\in \mathbb{D}$, we have F(u) and G(v) $\in \mathbb{D}$ and}
\begin{equation}\label{functional-calculus}
  \Gamma\left[F(u), \, G(v) \right] = F'(u) \,G'(v) \,  \Gamma[u, \, v] \; \;
  \widetilde{\mathbb{P}} \; a.s.; 
  \end{equation}
\item{The bilinear form $\mathcal{E}[u, \, v] = \frac{1}{2}
\widetilde{\mathbb{E}}\left[\Gamma[u,
      \, v]\right]$ is closed;}

We generally write $\Gamma[u]$ for $\Gamma[u,\, u]$.
\end{itemize}
\end{definizione}

With this definition $\mathcal{E}$ is a Dirichlet form and $\Gamma$ is the associated carr\'{e} du champ operator. The Hille Yosida theorem, see for instance Albeverio \cite{bib:Albeverio} and Fukushima et al. \cite{bib:Fukushima}, 
guarantees that it exists a semigroup, and, then, a generator $\mathcal{A}$ coherent with the Dirichlet form  
$\mathcal{E}$. This generator $( \mathcal{A}, \, \mathcal{D}  \mathcal{A})$  is a 
self-adjoint operator that satisfies, for $F \in \mathcal{C}^2$, $u \in \mathcal{D}  \mathcal{A}$ 
and $\Gamma[u] \in L^2(\widetilde{\mathbb{P}})$:
\begin{equation}\label{bias-chain-rule}
  \mathcal{A}\left[F(u) \right] = F'(u)\,  \mathcal{A}[u] + \frac{1}{2} F''(u)\,  \Gamma[u] \; \;
  \widetilde{\mathbb{P}} \; a.s.;
\end{equation}
moreover, it is a closed operator with respect to the graph norm.
We underline two important results related to this theory. First of all the concept of error structure is 
deeply related to statistics. As a matter of fact, the uncertainties on parameters come from a statistical estimation and Bouleau and Chorro \cite{bib:Bouleau-Chorro} have proved a connection between error theory using Dirichlet forms and Fischer information theory. 
Furthermore, the error structures have nice properties, in particular is possible to prove that the product of two or countably many error structures is an error structure, see Bouleau \cite{bib:Bouleau-erreur2}. 

The main drawback of the carr\'{e} du champ operator is its bi-linearity that makes computations awkward to perform. An easy way to overcome it is to introduce a new operator, the sharp.
We recall the definition of sharp operator associated with $\Gamma$, see 
Bouleau and Hirsch \cite{bib:Bouleau-Hirsch}, section II.6. for the proof.

\begin{proposizione}[Sharp operator]\label{prop:sharp}\hfill

Let $\left( \widetilde{\Omega}, \, \widetilde{\mathcal{F}}, \,
\widetilde{\mathbb{P}}, \, \mathbb{D}, \, \Gamma \right)$ be an error
structure and $\left( \widehat{\Omega}, \, \widehat{\mathcal{F}},
\, \widehat{\mathbb{P}} \right) $ a copy of the probability space
$\left( \widetilde{\Omega}, \, \widetilde{\mathcal{F}}, \,
\widetilde{\mathbb{P}}\right) $. We assume 
that the space $\mathbb{D}$ is separable, then there exists a
sharp operator $(\,)^{\#}$ with these three properties:

\begin{itemize}

\item{$\forall \, u \in \mathbb{D}$, $u^{\#} \in
L^2(\widetilde{\mathbb{P}} \times \widehat{\mathbb{P}})$;}

\item{$\forall \, u \in \mathbb{D}$, $\Gamma[u] =
\widehat{\mathbb{E}}\left[\left(u^{\#}\right)^2\right]$;}

\item{$\forall \, u \in \mathbb{D}^n$ and $F \in \mathcal{C}^1 \cap
Lip$,  $\left(F(u_1, \, ... \,, \, u_n)\right)^{\#} =
 \sum_{i=1}^n \left(\frac{\partial F}{\partial x_i}\circ u \right) u_i^{\#}  $.}

\end{itemize}
\end{proposizione}

The sharp operator is an useful tool to compute $\Gamma$ because the
sharp is linear whereas the carr\'{e} du champ is bilinear. 
In analogy with the classical approach of error theory we associate
the carr\'{e} du champ operator $\Gamma$ to the
normalized variance of the error, the sharp operator becomes a linear 
version of the standard deviation of the error.
Similarly, the generator describes the error biases after
normalization, for more details we refers to  Bouleau
\cite{bib:Bouleau-erreur} chapters III and V and
\cite{bib:Bouleau-erreur2}. 

In order to apply the result of error theory to our problem, we assume that
we can define an error structure for each parameter of our model, 
we denote them $a_i$, and that we suppose the following assumption held 
for each parameter.

%%%%%%%%%%%%%%%%%%%%%%

\begin{assumption}[Error theory framework]\label{assumption-1}\hfill

\begin{enumerate}
\item $a_i \in \mathcal{D} \mathcal{A}$, $\Gamma[a_i]$ and $\mathcal{A}[a_i]$ are known;

 \item the functions $x \mapsto \Gamma[a_i](x)$ and $x \mapsto \mathcal{A}[a_i](x)$ 
 belong to $L^1 \cap L^2$ are continuous at $x=a_i$ and not 
 vanishing\footnote{We use an abuse of language, since when we write
 $\Gamma[a_i]$ and $\mathcal{A}[a_i]$, we mean that $a_i$ is a random variable. 
 Afterwards, we compute $\Gamma[a_i]$ and $\mathcal{A}[a_i]$ on the point $x=a_i$,
 where $a_i$ denotes the estimated value of the parameter.};
\item the error structure admits  a sharp operator denoted $( \,)^{\#}$.
\end{enumerate}

\end{assumption}

We conclude with a direct corollary of this theory. The error theory using Dirichlet forms
restricts its analysis to the study of two first order errors propagation, i.e. the  bias and the
variance. This fact is justified by the knowledge on the parameters uncertainties, generally 
given by the Fischer information matrix, that is often bad. The study of the high orders is a very
hard problem for both mathematical and practical reasons. From mathematical point of view,
we have to study the chain rules of the high orders, like skewness and kurtosis, and to prove 
that the related operators are closed in a suitable space. However, the crucial problem 
remains to have fine estimations for the high orders of the uncertainties, this obstacle cannot be 
overcome easily. Therefore, we decide to restrict our study to the two first order. 
As a consequence, we implicitly work with gaussian random variables.   
Finally we can state:

\begin{remarque}[impact of uncertainty]\label{remark-expansion}\hfill

The impact of uncertainty on the parameter transforms a constant into
a gaussian distribution of the form
\begin{equation}
F(X) \sim F(x) + \epsilon \;
A[F(X)](X= x) + \sqrt{\epsilon} \;
\sqrt{\Gamma[F(X)](X= x)} \; G
\end{equation}
where G is a standard Gaussian variable and $\epsilon$ is a vanishing parameter.
\end{remarque}

In a more conservative way, i.e. if we suppose that the gaussian approximation is
too rough, we can apply the Chebyshev's inequality that state in our case:

\begin{proposizione}[Chebyshev's inequality]\label{prop:Cheb}\hfill

The random variable $F(X)$ verifies the following inequality for all $k\geq 1$ and supposing
that $\epsilon$ is a vanishing parameter.
\begin{equation}\label{cheb-equation}
\widetilde{\mathbb{P}}\left[F(X)-F(x) - A[F(X)](X= x) \geq k\,  \sqrt{\epsilon 
\; \Gamma[F(X)](X= x)}  \right] \leq \frac{1}{1+k^2}
\end{equation}
\end{proposizione} 

These results explain the role of the generator and the carr\'{e} du
champ operator, see Bouleau \cite{bib:Bouleau-MC} for a more general analysis.
The theoretical image is
perturbed due to the uncertainty on the parameter, this effect is
small, however  it produces not only a noise but also it alters the
mean.

\section{Diffusion model under perturbation}

We start with the classical Black Scholes model, afterward denoted BS, see Black and Scholes
\cite{bib:Black-Scholes}. Let
$\left(\Omega, \, \mathcal{F}, \, \mathbb{P}\right)$ be the historical
probability space and $B_t$ the associated Brownian motion,  the dynamic of the risky asset under historical
probability $\mathbb{P}$ is given by the following diffusion in accord with the model of Black and Scholes:
\begin{equation}\label{Black-Scholes}
dS_t  =  S_t \, \mu \, dt + S_t \, \sigma_0 \, dB_t 
\end{equation}
In this framework, the price of a European vanilla option is well known.
This model presents many advantages, in particular the pricing depends
only on volatility and we find closed forms for premium
and greeks of vanilla options. Unluckily, the BS model cannot
reproduce the market price of call options for all strikes at the
same volatility, this effect is called smile.
To take into account this phenomenon, we analyze two main extensions, the local and the stochastic volatility models, afterward denoted LV and SV, see for instance Dupire \cite{bib:Dupire}, Hull and White \cite{bib:Hull} and Heston \cite{bib:Heston}. 
In this two classes of models, the parameter $\sigma_0$ is replaced by a function $\sigma$ that depends on the time $t$, on the underlying $S_t$ and, in the case of SV models, on a random source. The stochastic differential equation verified by the price of the underlying is
\begin{equation}\label{SDE-SV}
dS_t = S_t \, \mu \, dt + S_t \, \sigma(t, \, S_t, \, \omega) \, dB_t 
\end{equation}
This class of model is large enough to include a lot of stochastic model currently used in 
finance, like SABR and fast mean reverting SV, see Hagan et al \cite{bib:Hagan} and Fouque et al \cite{bib:Fouque}. However, all diffusion with jumps model are excluded from our analysis, see 
for instance Cont and Tankov \cite{bib:Cont}.

Henceforth, we denote $(\Omega, \, \mathcal{F},\, \mathbb{P})$ the probability space where the Brownian motion is defined, $\mathcal{F}_t$ the standard filtration generated by the Brownian motion $B_t$ and $\mathbb{E}$ the expectation under the probability $\mathbb{P}$.   

Our goal is to analyze the sensitivities of these models with respect to the fluctuations on their parameters. As a matter of fact, all models depend on some parameters, generally few, that depend on the underlying and enable the calibration of the model. When we sort out a given value for a parameter, thanks to a calibration methodology, it is known with an uncertainty. This uncertainty can be estimated using statistical methods, like Fischer information, or computing the sensitivity of the model calibration. We underline that each method, both statistical and calibrative, yields parameters with uncertainties and these uncertainties have a random nature. Our goal is to analyze their impact on the management and the hedging of a contingent claim on the side of the seller that tries to escape any risk. 

We propose to consider a perturbation of the model given by SDE (\ref{SDE-SV})  by means of an
error structure on the volatility function $\sigma(t, \,S_t, \, \omega)$.

We make the following  hypotheses:

\begin{assumption}[Asset evolution and uncertainty impact]\label{assumption-fin}\hfill

\begin{enumerate}
\item{ the real market follows the SDE (\ref{SDE-SV})  with fixed but unknown parameters, i.e. the underlying follows the SDE (\ref{SDE-SV}) and this diffusion does not suffer the uncertainty on the parameters. The market is viable and complete, i.e. they are enough traded assets to guarantee that any contingent claims admits  an hedging portfolio;}

\item{the option seller knows that the underlying follows the SDE (\ref{SDE-SV}) but does not know the function $\sigma$, i.e. the values of the parameters;}

\item{ the option seller has to estimate the parameters of his model, so its
volatility contains intrinsic inaccuracies, we model this ambiguity
by means of an error structure; nonetheless we assume that the stock
price $S_t$ is not erroneous, but the price and the greeks of the option yes. 
The option seller evaluates the impact of this uncertainty on his profit and loss process, afterward denoted $P\&L$, and tries to compensate it modifying his prices}.
\end{enumerate}
\end{assumption}

\begin{assumption}[Uncertainty on volatility]\label{assumption-math}\hfill

\begin{enumerate}
\item{the uncertainty on the volatility function $\sigma(t, \,S_t, \, \omega)$ is described by an error structure $(\widetilde{\Omega}, \, \widetilde{\mathcal{F}}, \, \widetilde{\mathbb{P}}, \mathbb{D}, \, \Gamma)$ associated to the function $\sigma$;} 
\item{the uncertainty on $\sigma$ is independent of the remain of the model, i.e. of the probability space $(\Omega, \, \mathcal{F}, \, \mathbb{P})$   where the Brownian motion $B_t$ is defined;}
\item{the error structure and the function $\sigma$ verify the assumption  \ref{assumption-1}.}
\item{the function $\sigma(t, \, x, \, \omega)$ is square integrable and admits a series expansion, i.e. $\sigma(t, \, x, \, \omega) = \sum_i a_i \, \phi_i(t, \, x, \, \omega)$, where the functions $\phi_i$ belong to $C^{1, \, 2}$ on $(t,\, x)$ and are square integrable on the third variable. We assume that the uncertainty is carried by the coefficients $a_i$, i.e. we set the $a_i$ to be random and the error structure is given by the product  of the error structures on each $a_i$. We also suppose held the property 4.2. in Bouleau \cite{bib:Bouleau-erreur} chapter V page 84, thanks to that only a finite number of coefficients is non zero.}
\end{enumerate}
\end{assumption}

Before continuing, we discuss briefly the two previous assumptions. We have divided them depending on the nature of the assumption, more financial for the first one and mathematical for the second. Assumption \ref{assumption-math} sets the mathematical frame used in this paper, points 2 and 4 are assumed to make easier the proof of results and to simplify some computation, but they can be replaced with more general hypothesis. 

The main financial hypotheses are resumed in assumption \ref{assumption-fin}. The first point states that the underlying follows an exact SDE without perturbation. The uncertainty is added by the option seller when he searches to know the parameters of the SDE. This assumption is easy to 
understand in financial frame and avoids many problems on the mathematical framework. In particular when the volatility is uncertain the set of probabilities describing the whole class of possible probabilistic views is not  dominated. In this case, the classical approach used in mathematical finance cannot be followed, the main examples about this subject are Denis and Martini \cite{bib:Denis-1} and Denis and Kervarec \cite{bib:Denis-2}.   

We also add the following hypothesis on the class of contingent claim analyzed in this work.

\begin{assumption}[Contingent Claims]\label{assumption-payoff}\hfill

Let $\Phi$ be the payoff of the contingent claim. We assume that $\Phi$ depends only on $S_T$ and belongs to $C^2$. We also suppose that the two first derivatives are bounded.
\end{assumption}

It is plain that the previous assumption is very restrictive, since call and put does not verifies for instance. However, it is easy to define a series of payoff $\overline{\Phi}_i$ that dominates the payoff of a call (or a put) and  converges uniformly and to take the limit using the hypothesis of no free lunch with vanishing risk.

\subsection{Management of profit and loss process}

We consider  that the option seller uses his data and proprietary view to set 
the model parameters using an optimization procedure, see for instance Dupire \cite{bib:Dupire}  and Cont Tankov \cite{bib:Cont}. Given these parameters, the risk neutral-measure $\mathbb{Q}$ 
exists and is unique, this probability enable him to define the fair price of the option and the hedging strategy under the hypothesis that we knowledge the real values of the paprameters. 

We study the profit and loss process associated to this hedging portfolio.
We neglect, in this first work, that the option holder can sell the option, so we assume that he holds the contingent claim until the maturity. We also assume that all prices are denominated using the risk-free asset as numeraire.   
The profit and loss process at the maturity of the option seller is given by the price of the option plus the hedging strategy minus the final payoff that the seller has to pay to the holder, so we have
\begin{equation}\label{equation:P-and-L}
P\&L(T) = F(\varsigma, \, x, \, 0) + \int_0^T \Delta(\varsigma, \, S_t, \, t) dS_t - \Phi(S_T),
\end{equation}
where $F$ and $\Delta$ denote respectively  the price  of the option and its first derivative with respect to the underlying. The key remark about the P\&L equation   
(\ref{equation:P-and-L}) is that the price and the delta depend on the volatility process, estimated by the trader, $\varsigma$ that is not the real market volatility $\sigma$, i.e. the parameters used by the option seller to compute the price and the hedging position are their estimated parameters that are different from the real parameters in diffusion (\ref{SDE-SV}). It is plain that a good calibration produces a set of parameters very close to the real values but the uncertainty remains, as an example the variance of a unbiased statistical estimator is at least high as the Cramer Rao bound.  Moreover, thanks to a result of Bouleau and Chorro \cite{bib:Bouleau-Chorro}, $\Gamma$ is
equal the inverse of the Fisher information matrix, so the use of error theory using Dirichlet forms  
is efficient since it exploits all information on data given the fact that the Cramer Rao bound is verified.
  
We concentrate our attention on the law of the P\&L at the maturity. In absence of uncertainty the random value $P\& L(T)$ is equal to zero almost surely, i.e. the option can be exactly hedged. However, the exact hedging strategy cannot be performed since the option seller does not know the parameters value. Therefore, $P\&L(T)$ is a random variable and we have the following remark.  

\begin{remarque}[Random sources]\hfill

The value of profit and loss process at the maturity is a random variable depending on two random
sources:

\begin{itemize}
    \item First of all, the stochastic "real" model $(\Omega, \, \mathcal{F},\, \mathbb{P})$ since the trader cannot use
    the correct hedging portfolio.
    \item Second, the space $(\widetilde{\Omega}, \, \widetilde{\mathcal{F}},\, \widetilde{\mathbb{P}})$, i.e.  the stochastic process $\varsigma$, that depends
    on a random component independent to the brownian motion $B_t$.
\end{itemize}
\end{remarque}

\begin{remarque}[Role of historical probability]\hfill

The profit and loss process must be studied on historical probability
$\mathbb{P}$. As a matter of fact, the risk neutral probability $\mathbb{Q}$ can be used if and only if the market is complete, i.e. if all contingent claims are attainable. In our case, the option seller does not know the real diffusion of the underlying, so the law of the $P\&L$ is not degenerate. The main impact is that the drift of the SDE (\ref{SDE-SV}) plays a crucial role and that the second term in $P\&L$ process is not a martingale. This fact complicates the computations on this paper. However, the important role of the drift $\mu$ in asset pricing when the market is incomplete is emphasized in literature, let us mention Avelaneda et al. \cite{bib:Avelaneda}, Karatzas et al. \cite{bib:Karatzas} and Lyons \cite{bib:Lyons}. 
\end{remarque}

The two previous remarks show that the computation of the law of the $P\&L$ process is an hard problem. Moreover, an essential requisite lacks to define the law of the $P\&L$ since the $P\&L$ depends on $\varsigma$ and the law of this process is poorly known since it depends on a calibration methodology. We assume, implicitly on assumption \ref{assumption-math}, that the option seller can estimate the mean value of his parameters and their variance, i.e. the two first orders of the law of $\varsigma$. The knowledge of the high orders is a statistical problem often too hard to solve. 

Given the knowledge of the two first orders of the law of $\varsigma$, we can estimate the two first orders of the law of $P\&L$ at  best. This remark justifies our recourse to the error theory using Dirichlet forms, see Bouleau \cite{bib:Bouleau-erreur}. More precisely, we have remarked that the profit and loss process depends on two random sources, one is the stochastic process $\varsigma$ depends on some parameters that are random variables on the space $(\widetilde{\Omega}, \, \widetilde{\mathcal{F}}, \, \widetilde{\mathbb{P}})$, the second is the Brownian motion $B_t$ that live on the probability space $(\Omega, \, \mathcal{F}, \, \mathbb{P})$. We can define accurately the law of $P\&L$ with respect to the random source on the probability space $\mathbb{P}$ but 
we can just estimate the two first orders of the dependency with respect to 
$\widetilde{\mathbb{P}}$.

In order to analyze the law of P\&L process, it is sufficient to
study the $\mathbb{P}$-expectation on a class of regular test functions $h(P \& L)$ and
the error on them using error theory. In practice we will compute the bias and the variance of $\mathbb{E}[h(P\&L)]$, where $\mathbb{E}$ denotes the expectation under the probability $\mathbb{P}$. To simplify our results, we suppose moreover the following assumption held.

\begin{assumption}[Expansion Approach]\label{assumption-EA}\hfill

We assume that the volatility estimated by the option seller is unbiased, i.e.  we
compute the value of our expectation for $\varsigma= \sigma$.
\end{assumption}

The previous assumption involves in particular that the $P\&L$ is worth zero $\mathbb{P}$-almost surely, so this hypothesis enable us to simplify our results in two ways. First of all, we don't need to distinguish the two volatilities $\sigma$ and $\varsigma$, we are interested only on the stochastic correction, in this sense we have an expansion approach. Secondly under $\mathbb{P}$ the random variable $P\& L(T)$ is a constant.
Then we can state the following theorem.

 \begin{teorema}[Law of the profit and loss process]\label{theorem-main}\hfill
 
 Under assumptions  \ref{assumption-fin}, \ref{assumption-math}, \ref{assumption-payoff} and 
 \ref{assumption-EA},
we have the following bias and variance:

\begin{eqnarray}
A[\mathbb{E}[h(P \& L)]] & = &  h'(0) \, \Lambda_1 (\sigma) +
\frac{1}{2} \, h''(0) \, \Lambda_2(\sigma) \label{A-P-and-L} \\
    \Gamma \left[\mathbb{E}\left[h\left( P\&L \right)\right]
    \right] & = &  [h'(0)]^2 \; \Psi (\sigma) \label{gamma-P-and-L},
\end{eqnarray}
where
\begin{eqnarray*}
\Lambda_1(\sigma) & = &  \sum_i  \frac{\partial F}{\partial \sigma(\phi_i)} 
(\varsigma, x,\, 0) \;\phi_i(0,\, x) \; \mathcal{A}[a_i](a_i) + \frac{1}{2}  \sum_{i} \frac{\partial^2 F}{\partial [\sigma(\phi_i)]^2 
}  (\varsigma, x,\, 0) \;\phi^2_i(0,\, x)  \; \Gamma[a_i](a_i) \\
&  & \displaystyle + \sum_i \int_0^T \mu \, \mathbb{E} \left[ \frac{\partial \Delta}{\partial 
\sigma(\phi_i)} (\varsigma, S_t,\, t) \;\phi_i(t,\, S_t,\, \omega) \, S_t \right]\, dt \; 
\, \mathcal{A}[a_i](a_i)   \\
 &  &  \displaystyle  + \frac{1}{2}  \sum_{i}  \int_0^T  \mu \, \mathbb{E}\left[
\frac{\partial^2 \Delta}{\partial [\sigma(\phi_i)]^2 
}  (\varsigma, S_t,\, t) \;\phi^2_i(t,\, S_t,\, \omega)\, S_t \right] \, dt \; \, \Gamma[a_i](a_i) \\ 
\Lambda_2(\sigma) & =& \displaystyle 
  \sum_i   \Gamma \left[ a_i \right] \left( a_i \right) \; \,  \left\{   \int_0^T \mathbb{E} \left[  \left(
\frac{\partial \Delta}{\partial \sigma(\phi_i)}(\varsigma, \, S_s, \, s) \, \phi_i(s,\, S_s,\, \omega)
\right)^2 \; S_s^2 \; \sigma^2 \right] ds   \right. \\
& & \displaystyle + \left. \mathbb{E} \left[
 \left( \frac{\partial F}{\partial \sigma(\phi_i)}(\varsigma,\, x, \, 0)\, \phi_i(0, \, x) + 
 \mu\, \int_0^T  \frac{\partial \Delta}{\partial \sigma(\phi_i)}(\varsigma, \, S_s, \, s)  \,  
 \phi_i(s,\, S_s, \, \omega) \, S_s \, ds 
        \right)^2 \right] \right\} \\
\Psi(\sigma) & = & \sum_i  \left\{\frac{\partial
F}{\partial \sigma (\phi_i)}(\sigma, x,
    0) \, \phi_i(0, \, x) + \mu \,  \int_0^T \mathbb{E} \left[ \frac{\partial \Delta}{\partial \sigma(\phi_i)}(\sigma, \, S_s, s) \, S_s \, \phi_i(s, \, S_s, \, \omega) \right] ds \right\}^2 \Gamma[a_i](a_i),
\end{eqnarray*}
where $\frac{\partial }{\partial \sigma(\phi_i)}$ denotes the Gateaux derivative with respect to a variation of the volatility along the component $\phi_i$.
Moreover, we have the following truncated expansion
\begin{equation}\label{equation:rel_h}
  \mathbb{E}\left[h(P\&L)\right] \sim h(0)+ \epsilon \, h'(0) \, \Lambda_1(\sigma) + \epsilon
   \, \frac{1}{2}h''(0) \,\Lambda_2(\sigma) + \sqrt{\epsilon \,
\left[h'(0) \right]^2 \,
  \Psi (\sigma)} \; \widetilde{G},
  \end{equation}
where $\widetilde{G}$ is a standard gaussian random variable. Otherwise, we have the following Chebyshev's inequality
\begin{equation}
\widetilde{\mathbb{P}}\left[  \mathbb{E}\left[h(P\&L)\right] - 
h(0)- \epsilon \; \mathcal{A} \left[  \mathbb{E}\left[h(P\&L)\right]^{ }\right]  \geq k \,  \sqrt{\epsilon \;
   \Gamma \left[\mathbb{E}\left[h\left( P\&L \right)\right]
    \right]   }  \right] \leq \frac{1}{1+k^2}.
\end{equation}

\end{teorema}

\begin{proof}
We start with the study of the variance. Thanks to assumption \ref{assumption-1}, included in assumption \ref{assumption-math}, the operator $\Gamma$ admits a sharp operator, in particular in each sub-error structure associated to each parameter $a_i$ it exists a sharp operator denoted $()^{\#}$. Moreover, the error structure, defined in assumption \ref{assumption-math}, is the product of the sub-error structures on each parameter, so the different sub-error structures are independent among themselves. 

We now study the sharp of $\mathbb{E}\left[h\left( P\&L \right)\right]$. Thanks to the linearity of the sharp operator and the expectation and the smoothness of the test function $h$ we have
\begin{equation}\label{eqn:proof-gamma}
\begin{array}{rcl}
 \displaystyle \left(\mathbb{E}\left[h\left( P\&L \right)\right]\right)^{\#} & = & 
  \displaystyle   \mathbb{E}\left[h'\left(P\&L \right) \; (P\&L)^{\#} \right]  \\
    & =  & \displaystyle \mathbb{E}\left[h'\left(P\&L \right) \;
    \left\{ F(\varsigma, \, x, \, 0) + \int_0^T \Delta(\varsigma, \, S_t, \, t) dS_t - \Phi(S_T) \right\}^{\#}\right].
    \end{array}
\end{equation}
Using the linearity of the sharp on braces, we remark that $(\Phi(S_T))^{\#}=0$ since the payoff is independent on the volatility estimated by the option seller. The same argument and the closedness of the sharp enable us to write
$$
\left( \int_0^T \Delta(\varsigma, \, S_t, \, t) \, dS_t\right)^{\#} =  \int_0^T \left(\Delta(\varsigma, \, S_t, \, t) \right)^{\#} \,dS_t,
$$
where we have used the first point of assumption \ref{assumption-fin}.
Now, we recall that the price and the hedging position on the underlying are smooth functions with respect to the volatility since the SDE  (\ref{SDE-SV}) has a gaussian kernel and thanks to assumption \ref{assumption-payoff}. Therefore, we can take the Gateaux-derivatives of $F$ and $\Delta$ with respect to the volatility change in the direction of the functions $\phi_i$, denoted $\frac{\partial}{\partial \sigma(\phi_i)}$ and we have     
\begin{equation}\label{eqn:proof-sharp}
\begin{array}{rcl}
\displaystyle F(\varsigma,\, x, \, 0)^{\#} &= & \displaystyle  \sum_i  \frac{\partial F}{\partial \sigma(\phi_i)} (\varsigma, x,\, 0) \;\phi_i(0,\, x,\, \omega) \; a_i^{\#} \\
\displaystyle  \left(\int_0^T \Delta(\varsigma, \, S_t, \, t) \, dS_t\right)^{\#} &= &
\displaystyle \sum_i a_i^{\#} \; \int_0^T
 \frac{\partial \Delta}{\partial \sigma(\phi_i)}(\varsigma, \, S_t, \, t) \; \phi_i(t,\, S_t, \, \omega) \, dS_t  
\end{array}
\end{equation}
where we have make use of the expansion of the volatility given by assumption \ref{assumption-math}, as well the linearity of the sharp operator, see proposition \ref{prop:sharp}.
We take the $\mathbb{P}$-expectation of the two right sides. We remark that $\phi_i(0,\, x,\, \omega)$ must be $\mathcal{F}_0$-measurable, so is independent of $\omega$ and we denote it $\phi_i(0,\, x)$. 
In the second equation the integral is split into a Lebesque and a stochastic one. Thanks to the assumption \ref{assumption-payoff}, the stochastic integral is a martingale and the expectation is worth zero. Thanks to the Fubini theorem we can exchange the Lebesque integral with the expectation. At the end, we find
\begin{eqnarray*}
\mathbb{E}\left[F(\varsigma,\, x, \, 0)^{\#} \right] &= & \sum_i  \frac{\partial F}{\partial \sigma(\phi_i)} (\varsigma, x,\, 0) \phi_i(0,\, x) \; a_i^{\#} \\
\mathbb{E}\left[  \left(\int_0^T \Delta(\varsigma, \, S_t, \, t) \, dS_t\right)^{\#} \right] 
&= & \sum_i a_i^{\#} \; \mu\, \int_0^T \mathbb{E}\left[ \frac{\partial \Delta}{\partial \sigma(\phi_i)}(\varsigma, \, S_t, \, t) \, S_t \,  \phi_i(t,\, S_t, \, \omega) \right] \, dt .
\end{eqnarray*}
Now, we come back to the equation (\ref{eqn:proof-gamma}), the first term into the expectation is $h'(P\&L)$. The law of the profit and loss process is worth zero almost surely if we suppose that $\varsigma= \sigma$. Using the assumption \ref{assumption-EA}, we can therefore replace $h'(P\&L)$ with $h'(0)$.   

We conclude the proof on the relation for the quadratic error $\Gamma\left[ \; \mathbb{E}\left[h\left( P\&L \right)\right]\;  \right]$  using the independence between the error structure on each parameter $a_i$ and the definition of the sharp operator, so we find the relation (\ref{gamma-P-and-L}).

The study of the bias is more complicated, we start applying the linearity and the closedness of  the bias operator:
\begin{equation}\label{eqn:proof-bias}
 \mathcal{A} \left[\mathbb{E}\left[h(P \& L)\right]\right] = 
 \mathbb{E}\left[ \mathcal{A}[h(P \& L)]\right] =
  \mathbb{E} \left[ h'(P \& L) \, \mathcal{A}\left[ P \& L  \right] + \frac{1}{2} h^{\prime
    \prime}(P \& L)  \, \Gamma \left[P \& L\right]\right]
  \end{equation}
where, in the second identity, we have exploited the smoothness of $h$ and the chain rule of $\mathcal{A}$. We study the two terms of the previous identity separately. We start studying the term depending on the quadratic error. 
Thanks to relations (\ref{eqn:proof-sharp}) and the property of sharp, see proposition 
\ref{prop:sharp}, we have 
$$
\Gamma  \left[P \& L \right] = \sum_i   \left\{ \frac{\partial F}{\partial \sigma(\phi_i)} 
(\varsigma, x,\, 0) \;\phi_i(0,\, x,\, \omega)  + \int_0^T
 \frac{\partial \Delta}{\partial \sigma(\phi_i)}(\varsigma, \, S_t, \, t) \; \phi_i(t,\, S_t, \, \omega) \, dS_t
\right\}^2 \; \Gamma[a_i] (a_i).
 $$
 We remark that the integral into the brackets can be split into a Lebesque and a stochastic integral. We can now compute the expectation under $\mathbb{P}$ and we find
\begin{equation}\label{eqn:Gamma-proof}
\begin{array}{l}
 \displaystyle   \mathbb{E}\left[\Gamma \left[P \& L \right]\right]  =  \displaystyle 
  \sum_i   \Gamma \left[ a_i \right] \left( a_i \right) \; \,  \left\{   \int_0^T \mathbb{E} \left[  \left(
\frac{\partial \Delta}{\partial \sigma(\phi_i)}(\varsigma, \, S_s, \, s) \, \phi_i(s,\, S_s,\, \omega)
\right)^2 \; S_s^2 \; \sigma^2 \right] ds   \right. \\
\hspace{0.5cm} \displaystyle + \left. \mathbb{E} \left[
 \left( \frac{\partial F}{\partial \sigma(\phi_i)}(\varsigma,\, x, \, 0)\, \phi_i(0, \, x) + 
 \mu\, \int_0^T  \frac{\partial \Delta}{\partial \sigma(\phi_i)}(\varsigma, \, S_s, \, s)  \,  
 \phi_i(s,\, S_s, \, \omega) \, S_s \, ds 
        \right)^2 \right]
 \right\}
\end{array}
\end{equation}
where we have make use of the properties of the Ito-integrals. We remark in particular that 
$\mathbb{E}\left[\Gamma \left[P \& L \right]\right] \neq \Gamma \left[\mathbb{E} \left[P \& L \right]\right]$, since the operator $\Gamma$ is bilinear.

Now we study the term $\mathcal{A}[P\&L]$, we apply the linearity of this operator and we find
$$
\mathcal{A}[P\&L] = \mathcal{A}\left[F(\varsigma,\, x, \,0)\right] + \mathcal{A}\left[ \int_0^T \Delta(\varsigma,\, S_t, \,t)\, dS_t   \right] -  \mathcal{A}\left[ \Phi(S_T)\right].
$$
The last term is worth zero, since the final payoff is completely defined by $S_T$ and does not depend on the volatility estimated by the option seller, see assumptions \ref{assumption-fin} and \ref{assumption-payoff}. Thanks to the same assumptions and the closedness of the bias operator $\mathcal{A}$, we have
$$
\mathcal{A}\left[ \int_0^T \Delta(\varsigma,\, S_t, \,t)\, dS_t   \right] =  \int_0^T 
\mathcal{A}\left[\Delta(\varsigma,\, S_t, \,t) \right]\, dS_t.
$$
Thanks to the same argument used for the quadratic error, we can take the Gateaux-derivatives of $F$ and $\Delta$ with respect to a variation of the volatility along the component $\phi_i$ and using the bias chain rule (\ref{bias-chain-rule}), we have  
\begin{equation*}
\begin{array}{rcl}
\displaystyle \mathcal{A}\left[F(\varsigma,\, x, \,0)\right]   & = & \displaystyle
 \sum_i  \frac{\partial F}{\partial \sigma(\phi_i)} 
(\varsigma, x,\, 0) \;\phi_i(0,\, x,\, \omega) \; \mathcal{A}[a_i](a_i) \\
& & \displaystyle + \frac{1}{2}  \sum_{i} \frac{\partial^2 F}{\partial [\sigma(\phi_i)]^2 
}  (\varsigma, x,\, 0) \;\phi^2_i(0,\, x,\, \omega)  \; \Gamma[a_i](a_i) \\
\displaystyle  \int_0^T  \mathcal{A}\left[\Delta(\varsigma,\, S_t, \,t) \right]\, dS_t  
& = & \displaystyle \sum_i \left[ \int_0^T \frac{\partial \Delta}{\partial \sigma(\phi_i)} 
(\varsigma, S_t,\, t) \;\phi_i(t,\, S_t,\, \omega) \, dS_t \right]\; \, \mathcal{A}[a_i](a_i)   \\
 &  &  \displaystyle  + \frac{1}{2}  \sum_{i}  \left[\int_0^T 
\frac{\partial^2 \Delta}{\partial [\sigma(\phi_i)]^2 
}  (\varsigma, S_t,\, t) \;\phi^2_i(t,\, S_t,\, \omega)\, dS_t \right] \; \, \Gamma[a_i](a_i)
\end{array}
\end{equation*}
where we have used that the sum over the index i has a finite number of terms thanks to the point 4 of the assumption \ref{assumption-math}. We can take the expectation and, thanks to assumptions \ref{assumption-math} and \ref{assumption-payoff}, the stochastic integrals are martingales. Therefore, we have
\begin{equation}\label{eqn:Bais-proof}
\begin{array}{rcl}
\displaystyle \mathbb{E}\left[\mathcal{A}[P \& L] \right] & = & \displaystyle 
    \sum_i  \frac{\partial F}{\partial \sigma(\phi_i)} 
(\varsigma, x,\, 0) \;\phi_i(0,\, x) \; \mathcal{A}[a_i](a_i) \\
&  & \displaystyle + \sum_i \int_0^T \mu \, \mathbb{E} \left[ \frac{\partial \Delta}{\partial 
\sigma(\phi_i)} (\varsigma, S_t,\, t) \;\phi_i(t,\, S_t,\, \omega) \, S_t \right]\, dt \; 
\, \mathcal{A}[a_i](a_i)   \\
& & \displaystyle + \frac{1}{2}  \sum_{i} \frac{\partial^2 F}{\partial [\sigma(\phi_i)]^2 
}  (\varsigma, x,\, 0) \;\phi^2_i(0,\, x)  \; \Gamma[a_i](a_i) \\ 
 &  &  \displaystyle  + \frac{1}{2}  \sum_{i}  \int_0^T  \mu \, \mathbb{E}\left[
\frac{\partial^2 \Delta}{\partial [\sigma(\phi_i)]^2 
}  (\varsigma, S_t,\, t) \;\phi^2_i(t,\, S_t,\, \omega)\, S_t \right] \, dt \; \, \Gamma[a_i](a_i).
\end{array}
\end{equation}
We consider therefore the relation (\ref{eqn:proof-bias}). Thanks to assumption 
\ref{assumption-EA}, we have $h'(P\&L)=h'(0)$ and $h^{\prime \prime}(P\& L) = h^{\prime \prime}(0)$. We conclude the computation using relations (\ref{eqn:Gamma-proof}) and (\ref{eqn:Bais-proof}). We find the relation  (\ref{A-P-and-L}).
The proof ends with the truncated expansion that is a consequence of remark \ref{remark-expansion} and proposition \ref{prop:Cheb}, see Bouleau
\cite{bib:Bouleau-MC} and \cite{bib:Bouleau-erreur4}.

\end{proof}

\subsection{Option Pricing}

In order to interpret this result in finance, we consider that the option seller 
knows the presence of errors in his procedure and wants to neutralize this effect.
It is plain that the option seller does not control the uncertainties, so 
the risk related to the space $(\widetilde{\Omega}, \, \widetilde{\mathcal{F}}, 
\, \widetilde{\mathbb{P}})$, i.e. to the uncertainty on the parameters, can not be hedged.
Furthermore, the relation  (\ref{equation:rel_h}) of the main theorem 
\ref{theorem-main} involves that it is not possible to bound this risk, i.e. to propose a super-hedging strategy, see Avelaneda \cite{bib:Avelaneda} and Lyons \cite{bib:Lyons}. Indeed, if we compute the super-hedging price of a contingent claim, we find a  very high buy-price and a very small sell-price, i.e. the bid-ask spread becomes to large compared with the market spreads.  

This result rises directly from the approach followed. 
We consider a very large class of model and any restrictive hypothesis on the law of the parameters uncertainties. The only strong hypothesis is on their magnitude, we have supposed that the uncertainty are small compared with the estimated values. We exploit further 
this property. Thanks to that, the probability that the random variable $\mathbb{E}[h(P\&L)]$ 
takes values far from the mean is very small and becomes negligible if this distance is big compared with the proper length $\sqrt{\epsilon [h'(0)]^2 \Psi(\sigma)}$.
Therefore, we introduce the following principle to define the price of a contingent claim under this uncertainty model.

\begin{principle}[Asset pricing under uncertainty]\label{principle-1}\hfill

The asset seller  fixes a tolerable risk probability $\alpha < 0.5$ and
accepts to sell the option at any price $F_{\text{sell}}$, such that 
\begin{equation}\label{eqn-principe-1}
\widetilde{\mathbb{P}}\left\{  F_{\text{sell}} - F+ \mathbb{E}\left[ P \& L\right]  < 0 \right\} \leq \alpha
\end{equation}
where $F$ denotes the cost of the hedging strategy, that is the theoretical price of the option without uncertainty, i.e. $F = \mathbb{E}[\Phi(S_T)]$. 
\end{principle}

Before applying this principle to our analysis, we discuss the financial implications of this principle.
The first remark is that this principle does not give a single price but an half-line of possible selling price. Indeed if $X$ is a selling price and $Y>X$ then 
$$
\widetilde{\mathbb{P}}\left\{  Y - F+ \mathbb{E}\left[P \& L\right]  < 0 \right\} < \widetilde{\mathbb{P}}\left\{ X - F+ \mathbb{E}\left[P \& L\right]  < 0 \right\} \leq \alpha,
$$
so $Y$ is a selling price too.
We call ask price, denoted $F_{\text{ask}}$, the infimum of the set of all selling option. 
We analyze the event $\{X - F+ \mathbb{E}\left[P \& L\right]  < 0\}$. On this event, if the seller option has sold the option at the price $X$, he loses money at maturity, since he has to pay $F$ to buy the hedging portfolio and the noise on this hedging strategy conducts him to lose money. However,  the probability of this event is smaller than $\alpha$. 

Therefore, the asset pricing principle says that the option seller accepts to sell the option at a price $X$ if this price is high enough to guarantee that he loses money with a probability smaller than $\alpha$. It is plain that the parameter $\alpha$ depends on the risk aversion of the option seller. However, if he is too risk-adverse, he proposes their contingent claims at a too high price, then the buyers can find other traders that offer the same options at a lower price. We remark that there is a 
likeness between this principle and the hypothesis testing in statistics, indeed $\alpha$ can be interpreted as the critical probability of the error of the first kind, i.e. "lose money in the contract" in our case. The error of the second kind is to propose a too high price such that the contract is not signed, i.e. "lose the opportunity to make money". We will study the problem of the optimal proposal price in a next paper.  
We now assume that the tolerable risk probability $\alpha$ is fixed and that the option seller chooses the smaller price consistent with him tolerable risk probability.

We also remark that the principle \ref{principle-1} defines the purchase price too. As a matter of fact, if a trader accept to buy an option, he has to take a negative position on the hedging portfolio to cover them, i.e. he has to follow the opposite of the hedging strategy. In accord with the principle \ref{principle-1}, the option buyer accepts to buy the option at any price $F_{\text{buy}}$ such that
$$
\widetilde{\mathbb{P}}\left\{ - F_{\text{buy}} + F- \mathbb{E}\left[ P \& L\right]  < 0 \right\} \leq \alpha,
$$
that is the relation (\ref{eqn-principe-1}) where we have changed all signs in order to consider the short position. the previous relation can be rewritten as 
$$
\widetilde{\mathbb{P}}\left\{ F_{\text{buy}} - F + \mathbb{E}\left[ P \& L\right]  > 0 \right\} \leq \alpha.
$$
The previous relation and the relation (\ref{eqn-principe-1}) lead us to the following remark:

\begin{remarque}[Bid-Ask Spread]\hfill

If all option traders on the market are risk adverse, then they are a difference between the best purchase price, denoted $F_{\text{bid}}$ and the best seller price $F_{\text{ask}}$. 
\end{remarque} 
This remark results from the fact that if the trader are risk adverse the tolerable risk probability must be smaller that $0.5$, then $F_{\text{bid}} < F + \mathbb{E}^{\widetilde{\mathbb{P}}} 
\left[ \, \mathbb{E}[P\&L] \right] < F_{\text{ask}}$. 

The previous principle and the theorem \ref{theorem-main} have the following immediate consequence:

\begin{proposizione}[Option prices]\hfill

We assume all hypotheses of theorem  \ref{theorem-main} held. We suppose that the option seller follows the principle \ref{principle-1}, then he accepts to sell the option at any price bigger than 
\begin{displaymath}
F_\text{ask} = F + \epsilon \; A
\left[\mathbb{E}[h(P \& L)] \right]  + \sqrt{\epsilon \; \Gamma
\left[\mathbb{E}[h(P \& L)]\right]}\; \mathcal{N}_{1-\alpha}
\end{displaymath}
where $\mathcal{N}_{1-\alpha}$ is the $(1-\alpha)$-quantile of the reduced
normal law, or a function given by Chebychev's inequality (\ref{cheb-equation}) in the conservative case. Likewise, the option buyer accepts to buy the option at any
price lower than
\begin{displaymath}
F_\text{bid} = F + \epsilon \; A
\left[\mathbb{E}[h(P \& L)]\right] + \sqrt{\epsilon \; \Gamma
\left[\mathbb{E}[h(P \& L)]\right]} \; \mathcal{N}_{\alpha}
\end{displaymath}

\end{proposizione}

%\begin{figure}[h!]
%  \begin{center}
%    \epsfxsize=13cm
%    $$
%    \epsfbox{error.eps}
%    $$
%    \caption{Impact of ambiguity: the Dirac distribution of price X becomes a continuous
%      distribution; the mean shifts of $\epsilon \; A[X]$ and the variance is $ \epsilon \;
%      \Gamma[X]$.} \label{fig:error}
%  \end{center}
%\end{figure}

\begin{remarque}[Mid price and Bid-Ask spread]\hfill

We remark that the two previous prices are symmetric, since
$\mathcal{N}_{\alpha} + \mathcal{N}_{1- \alpha}= 0$; therefore  we have
\begin{eqnarray}
\frac{F_{ask}+F_{bid}}{2} & = & F + \epsilon \; A
\left[\mathbb{E}[h(P \& L)]\right] \\
\frac{F_{ask}-F_{bid}}{2} & = & \sqrt{ \epsilon \; \Gamma
\left[\mathbb{E}[h(P \& L)]\right]} \, \mathcal{N}_{1-\alpha}
\end{eqnarray}
We emphasize that with our model we can reproduce a bid-ask spread
and we can associate its width to the trader's risk aversion (the
probability $\alpha$) and the volatility uncertainty (the term
$\sqrt{\epsilon \; \Gamma \left[\mathbb{E}[h(P \& L)]\right]}$).
Another interesting point it that the mid-price does not depend on 
the tolerable risk $\alpha$

\end{remarque}

\section{Example: log-normal diffusion}

In this section, we give an example of the previous results. In particular,
we consider the log-normal diffusion, i.e. the Black Scholes model \cite{bib:Black-Scholes}.
The underlying follows the SDE  (\ref{Black-Scholes}). In this case, the volatility 
is a parameter, so we replace them by the same parameter multiplied by the identity function.
In this model, the assumption \ref{assumption-math} is simplified, since the series expansion of point 4 is replaced by a unique function, i.e. the identity, while the points 1 and 3 can be replaced by the more general Hamza condition, see Bouleau and Hirsch \cite{bib:Bouleau-Hirsch}. 
We concentrate our analysis on call option, it is plain that a call option does not verify the assumption \ref{assumption-payoff}, but we have a closed form for all greeks in the case of log-normal diffusion and we can check that the Black Scholes pricing formula is $C^2$ at each time t strictly smaller than the maturity T and the vega, i.e. the derivative w.r.t. the volatility, vanishes when t goes to $T$. These properties guarantee that the theorem \ref{theorem-main} remain true  
even without assumption  \ref{assumption-payoff}. In order to simplify our numerical computation we assume that the drift $\mu=0$ under historical probability $\mathbb{P}$, the computations in the case $\mu \neq 0$ can be found on Regis and Scotti \cite{bib:Scotti-1}.

In this case, the theorem \ref{theorem-main} has the following corollary:

\begin{corollario}[Bid and Ask prices with log-normal diffusion]\label{pricing-BS}\hfill

If the underlying follows the Black Scholes SDE   (\ref{Black-Scholes}) without drift, then the bias and the variance of the profit and loss process that hedges a call option verify the following equations:

\begin{eqnarray}
\mathcal{A}[C(x, \, K, \, T)] & = & x \; \frac{e^{-\frac{1}{2} d_1^2}}{\sqrt{2
\pi}} \left\{ \left. A\left[\varsigma \sqrt{T}\right] \right|_{\varsigma=\sigma_0} +
 \frac{d_1 d_2}{2 \sigma_0 \sqrt{T}} \, \left. \Gamma\left[\varsigma \sqrt{T}\right]
 \right|_{\varsigma= \sigma_0} \right\} \\
 \Gamma[C(x, \, K, \, T)] & = & x^2 \; \frac{e^{- d_1^2}}{2 \pi} \, \left. 
 \Gamma\left[\varsigma \sqrt{T}\right]\right|_{\varsigma= \sigma_0}
\end{eqnarray}
where 
$$
d_1 = \frac{\ln x - \ln K + \frac{\sigma_0^2}{2} T
}{\sigma_0 \sqrt{T} }, \; \; d_2 = d_1 - \sigma_0 \sqrt{T},
$$
and $C(x,\, K, \, T)$ denotes the call of strike $K$, maturity $T$ and underlying that quotes $x$.
Moreover we have the following bid and ask prices:
\begin{eqnarray}
C_{\text{ask}}(x, \, K, \, T) &= & x\mathcal{N} (d_1) - K \mathcal{N}(d_2)   
 + \sqrt{\epsilon} \, x \; \frac{e^{-\frac{1}{2} d_1^2}}{\sqrt{2
\pi}} \, \sqrt{ \left. \Gamma\left[\varsigma \sqrt{T}\right]
 \right|_{\varsigma= \sigma_0} } \mathcal{N}_{1-\alpha} \\
& & + \epsilon\, x \;  
\frac{e^{-\frac{1}{2} d_1^2}}{\sqrt{2
\pi}} \left\{ \left.  \mathcal{A} \left[\varsigma \sqrt{T}\right] \right|_{\varsigma=\sigma_0} +
 \frac{d_1 d_2}{2 \sigma_0 \sqrt{T}} \, \left. \Gamma\left[\varsigma \sqrt{T}\right]
 \right|_{\varsigma= \sigma_0} \right\} \nonumber  \\ 
 & & \nonumber \\ 
C_{\text{bid}}(x, \, K, \, T) & = & x\mathcal{N} (d_1) - K \mathcal{N}(d_2) 
+ \sqrt{\epsilon} \, x \; \frac{e^{-\frac{1}{2} d_1^2}}{\sqrt{2
\pi}} \, \sqrt{ \left.\Gamma\left[\varsigma \sqrt{T}\right]
 \right|_{\varsigma= \sigma_0} } \mathcal{N}_{\alpha}  \\ 
& &  + \epsilon\, x \; 
\frac{e^{-\frac{1}{2} d_1^2}}{\sqrt{2
\pi}} \left\{ \left. \mathcal{A} \left[\varsigma \sqrt{T}\right] \right|_{\varsigma=\sigma_0} +
 \frac{d_1 d_2}{2 \sigma_0 \sqrt{T}} \, \left. \Gamma\left[\varsigma \sqrt{T}\right]
 \right|_{\varsigma= \sigma_0} \right\} \nonumber 
\end{eqnarray}

\end{corollario}

\begin{proof}
We know that the price of a call on Black Scholes model is given by
$ C(x, \, K, \, T)  = F(\sigma_0, \, x, \, 0) = x\mathcal{N} (d_1) - K \mathcal{N}(d_2) 
$, the two previous equation come directly from a computation of $\mathcal{A}[C(x, \, K, \, T)]$
and $\Gamma[C(x, \, K, \, T)]$ using theorem \ref{theorem-main}, for more detail see Scotti \cite{bib:Scotti-2}.
\end{proof}

\subsection{Analysis of uncertainty impact}

We now analyze the correction on the pricing formula due to the presence of an uncertainty on volatility. We have the following corollaries that can be proved with easy computations, we analyze only the mid-price, i.e. the average price between bid and ask. We also use the notation
$$
r_r^{BS}\left( \sigma_0 \sqrt{T}  \right) = 2 \, \sigma_0\, \sqrt{T} \, \frac{ \left. 
\mathcal{A} \left[\varsigma \sqrt{T}\right] \right|_{\varsigma=\sigma_0} }{
\left. \Gamma\left[\varsigma \sqrt{T}\right]
 \right|_{\varsigma= \sigma_0}}\, .
$$

\begin{corollario}[Delta and Gamma correction]\label{corollary-delta-gamma}\hfill

The bias on the call price due to the uncertainty on volatility verifies 
\begin{equation}
 \frac{\partial \mathcal{A}[C]}{\partial K}  
= \frac{d_1 \; \mathcal{A}[C]}{K \, \sigma_0 \,  \sqrt{T}} -
\frac{x}{2 \,  K \,  \sigma_0^2 \,  T} \frac{e^{-\frac{1}{2} d_1^2}}{\sqrt{2
\pi}} \,  (d_1 + d_2) \,  \left. \Gamma\left[\varsigma \sqrt{T}\right]
\right|_{\varsigma=\sigma_0}.
\end{equation}
Moreover, the bias and its first derivative w.r.t. K are positive (resp. negative) at the money 
if  $r_r^{BS}\left( \sigma_0 \sqrt{T}  \right) > (\, < \, ) \frac{1}{4} \sigma_0^2 \, T$. 
Finally the bias is convex on K at the money if and only if
$$
r_r^{BS}\left( \sigma_0 \sqrt{T}  \right) <   \frac{\sigma_0^4 \, T^2 + 4 \sigma_0^2 \, T + 32}{4 \sigma_0^2 \, T + 16}
$$
\end{corollario}

\begin{corollario}[Time evolution of the correction at the money]\hfill

The cross derivative of the bias two times with respect to the strike K and one time with respect to the maturity, i.e. $\frac{\partial^3 \mathcal{A} }{\partial K^2 \partial T}$ is positive if and only if 
$$
    r_r^{BS}\left(\sigma_0 \sqrt{T} \right)  > \frac{1}{4} \; \frac{ \sigma_0^2 \,  T
      \left(\sigma_0^2 \,  T -4\right)^2 +128 }{ 16 + \sigma_0^4 \, T^2 }
$$
\end{corollario}

For the numerical computations see Scotti \cite{bib:Scotti-2}.
The two previous corollaries have an interesting consequence. 

\begin{proposizione}[Smile on implied volatility]\hfill

It exists an interval $]a, \, b [$ such that $\forall\,  r_r^{BS}(\sigma_0 \, \sqrt{T}) \in ]a, \, b [$ the implied volatility computed using the price given by corollary \ref{pricing-BS} is a convex function around the money.  Furthermore, the second derivative of the implied volatility with respect to strike is a decreasing function of the maturity $T$.  That is to say that the implied volatility is smiled and this effect is more accentuated for short maturities than long one.
\end{proposizione}

\begin{proof}
We start fixing  $r_r^{BS}(\sigma_0 \, \sqrt{T}) = \frac{1}{4}\sigma_0^2 \, T$. Thanks to corollary 
\ref{corollary-delta-gamma}, we have that $\mathcal{A}[C]=0$ at the money
with its first derivative with respect to the strike. Besides the same corollary assures us 
that the bias $\mathcal{A}[C]$ is strictly convex around the money thanks the 
continuity of the derivatives.  For the same value of $r_r^{BS}(\sigma_0 \, \sqrt{T})$
 we easily check that the cross derivative of the bias two times w.r.t. $K$ 
 and one time w.r.t. T is negative at the money. 
We argue that the call price at the money is equal to the Black Scholes price 
 whereas around the money the call price in our model is bigger than the Black 
 Scholes price. 
Giving that the Black Scholes formula is used to find the implied, we conclude that the implied 
volatility is strictly convex around the money and that this convexity is more marked for short 
maturities. We now remark that the implied volatility belogns to $C^{\infty}$    
as function of the strike $K$ since it is a composition of infinite differentiable function. 
We conclude that, for each sufficiently small neighborhood of the money,  it exists an interval 
$]a, \, b [$,  including  $\frac{1}{4}\sigma_0^2 \, T$, such that the implied volatility remains 
strictly convex and with convexity decreasing with the maturity. 
\end{proof}

\end{document}